\documentclass{article}

\usepackage{arxiv}
\usepackage{url}
\usepackage{pifont,color}
\usepackage{pifont}
\usepackage{bbding}
\usepackage{graphicx}
\usepackage{amssymb}
\usepackage{epstopdf}
\usepackage{subfigure}
\usepackage{amsthm}
\usepackage[square, comma, sort&compress, numbers]{natbib}
\usepackage{amsmath}
\usepackage{caption2}

\def\MatrixFont{\bf}
\def\VectorFont{\bf}

\newcommand{\RNum}[1]{\uppercase\expandafter{\romannumeral #1\relax}}

\newcommand{\mA}{{\MatrixFont A}}
\newcommand{\mB}{{\MatrixFont B}}
\newcommand{\mC}{{\MatrixFont C}}

\newcommand{\mG}{{\MatrixFont G}}

\newcommand{\mI}{{\MatrixFont I}}

\newcommand{\mM}{{\MatrixFont M}}

\newcommand{\mR}{{\MatrixFont R}}

\newcommand{\mZ}{{\MatrixFont Z}}

\newcommand{\vx}{{\VectorFont x}}
\newcommand{\vy}{{\VectorFont y}}

\newtheorem{theorem}{Theorem}

\newtheorem{remark}{Remark}

\newtheorem{proposition}{Proposition}
\newtheorem{assumption}{Assumption}

\usepackage[ruled,vlined,linesnumbered]{algorithm2e}
\begin{document}
\title{An Unbiased Symmetric Matrix Estimator for Topology Inference under Partial Observability}

\author{Yupeng Chen\\
	College of Mathematics\\ 
	Sichuan University\\
	Chengdu, Sichuan 610064, China\\
	\texttt{chenyupeng@stu.scu.edu.cn}, \And  
	Zhiguo Wang\thanks{
		Corresponding
		author: Zhiguo Wang. The work of Zhiguo Wang is supported by the Fundamental Research Funds for the Central Universities.}\\
	College of Mathematics\\ 
	Sichuan University\\
	Chengdu, Sichuan 610064, China\\
	\texttt{wangzhiguo@scu.edu.cn}\And
	Xiaojing Shen\\
	College of Mathematics
	\\ Sichuan University
	\\Chengdu, Sichuan 610064, China\\
	\texttt{shenxj@scu.edu.cn}\\}

\maketitle
\begin{abstract}

Network topology inference is a fundamental problem in many applications of network science, such as locating the source of fake news and brain connectivity network detection. Many real-world situations suffer from a critical problem in which only a limited number of observed nodes are available. In this work, the problem of network topology inference under the framework of partial observability is considered.
Based on the vector autoregression model, we propose a novel unbiased estimator for symmetric network topology with Gaussian noise and the Laplacian combination rule. Theoretically, we prove that this estimator converges in probability to the network combination matrix. Furthermore, by utilizing the Gaussian mixture model algorithm, an effective algorithm called the network inference Gauss algorithm is developed to infer the network structure. Finally, compared with state-of-the-art methods, numerical experiments demonstrate that better performance is obtained in the case of small sample sizes when using the proposed algorithm.

\end{abstract}


\section{Introduction}

In network topology inference, the goal is to identify the network structure from the collected signals, which plays a vital role in many applications \cite{dong2019learning,yang2020network,mateos2019connecting}. In most existing research \cite{vosoughi2020large,zhang2014topology,sharma2019communication}, the network topology was inferred based on a complete set of observed data for all the entities of interest.

However, many real-world applications suffer from
the partial observability problem. For example, when investigating a social network with millions of members, only a limited number of observed nodes are available \cite{matta2020graph,wai2019community}. Thus, it is necessary to study the network topology inference based on partial observations, (see Fig. \ref{fig1}). Since the inference performance is affected by the missing observations \cite{omidshafiei2017deep,segarra2017network,santos2018consistent,santos2020local}, network topology
inference is challenging when the signals of the hidden nodes are not observed.

Several works formulated partially observed network topology inference as sparse or convex optimization problems \cite{chandrasekaran2010latent,buciulea2019network}. Recently, the problem of general network topology inference using Erd\H{o}s-R\'{e}nyi
(ER) models \cite{erdHos1959random} under partial observability was considered in \cite{matta2018consistent}. Moreover, some attractive estimators, such as the Granger, one-lag, and residual estimators, were proposed in \cite{matta2018consistent,matta2019graph}. Furthermore, it was proven that these estimators can be used to recover the subgraph of observed network nodes when the network size grows.

Inspired by the estimators proposed in \cite{matta2019graph},
we propose an efficient unbiased symmetric matrix estimator to infer the network topology under partial observability.
Compared with the work in \cite{matta2019graph}, we focus on the case in which the network is undirected and therefore the associated adjacency matrix is symmetric, and the network has a fixed size.
The major technical contributions of this work are summarized as follows.
\begin{itemize}
	\item An unbiased matrix estimator for symmetric combination matrices is proposed when incomplete graph signals are produced by the vector autoregression (VAR) model. Moreover, under the settings of Gaussian noise and Laplacian combination rule, we prove that the unbiased estimator converges to the network combination matrix in probability as the number of samples increases.
\item Experiments show the entries of
the unbiased matrix estimator follow a Gaussian mixture distribution. Thus, by exploiting this discovery, a node pair clustering based on the Gaussian mixture model (GMM) algorithm rather than the K-means algorithm \cite{hartigan1979algorithm} is used.
	\item  The proposed unbiased estimator is effective in obtaining a better initialization for the GMM. Combining the unbiased estimator with the node pair clustering, we propose an efficient network inference method called the network inference Gauss (NIG) algorithm. Simulations show that the NIG algorithm performs better than the current methods in the case of small sample sizes.
\end{itemize}
\begin{figure}[h]
	\centering
	\includegraphics[width=0.6\linewidth]{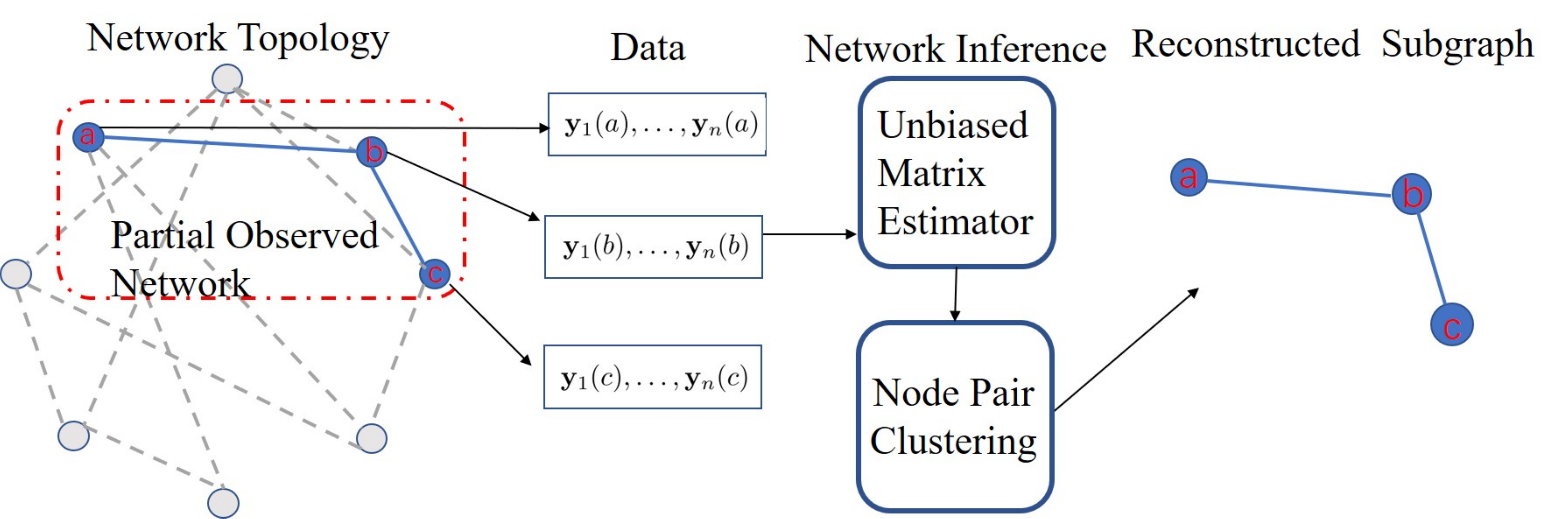}
	\caption{Illustration of the network inference under partial observability.}
	\label{fig1}
\end{figure}

\textbf{Notations:} The total number of network nodes is denoted by $N$, and the number of samples on each node is $n$.
For an $N\times N$ matrix $\mZ$, the submatrix spanning the rows and columns of $\mZ$ indexed by set $S\subset\{1,\dots,N\}$ is denoted by $\mZ_{S}.$ The fraction of the observed nodes is denoted as $\xi = \frac{|S|}{N},$ where $|S|$ represents the cardinality of the set $S$. Finally, $\mI$ denotes the identity matrix.

\section{Problem Formulation}
We consider an undirected
random network, such as ER and Barabási–Albert (BA) random graphs, with $N$ nodes.

At time $t$, each node $k$ collects the output signal $\vy_t(k)$ according to the following network diffusion process (i.e., running a VAR model) \cite{matta2020graph}
\begin{align}\label{va1}
	y_{t+1}(k)=\sum_{j=1}^{N}A_{kj}y_{t}(j)+\mu x_{t+1}(k),
\end{align}
where $x_{t}(k)$ is the input signal, $\mu\in(0,1)$ is the corresponding weighting coefficient and is assumed to be known.

Here, $x_{t}(k)$ for $t=1,\dots,n$ are assumed to be independent and identically distributed (i.i.d.) standard normal random variables, i.e., $x_{t}(k)\sim  \mathcal N (0,1).$ The $N\times N$ matrix with entries $A_{kj}$ is denoted as the network combination matrix $\mA$ that satisfies the following assumption.
	\begin{assumption}\label{Assu1}
	The network combination matrix $\mA$ is symmetric and obtained by applying the Laplacian
	combination rule \cite{matta2018consistent}: $$
	A_{k j}= \begin{cases}G_{kj}(1-\mu) \frac{\lambda}{d_{\max }}, & k \neq j, \\ (1-\mu)-\sum_{z \neq k} A_{k z}, & k=j,\end{cases}
	$$
	where $\lambda$ is a parameter with $\lambda\leq1$, $G_{kj}$ is the entry of matrix $\mG$, which is the adjacency matrix of the undirected random network, the degree of the $k$-th node is defined as $d_k=1+\sum_{z\neq k}G_{kz}$, and $d_{\max }$ is the maximum degree of matrix $\mG$.
	\end{assumption}
	Since the network combination matrix $\mA$ depends on the adjacency matrix $\mG$, it also represents the network topology structure. Specifically, $A_{kj}$ equals 0 if the $k$-th node disconnects from the $j$-th node, and $A_{kj}$ is a positive number $(1-\mu)\frac{\lambda} {d_{\max }}$, if the $k$-th node connects with the $j$-th node.

By stacking the input and output signals across the network at time $t$, equation \eqref{va1} is rewritten as
\begin{equation}\label{stack}
\vy_{t+1} = \mA\vy_{t} + \mu\vx_{t+1},
\end{equation}
where $\vy_{t}\triangleq[y_{t}(1),\dots,y_{t}(N)]$, $\vx_{t}\triangleq[x_{t}(1),\dots,x_{t}(N)]\in \mathbb{R}^N$. 

Assume that only a subset of the network nodes can be monitored. The monitored set is denoted as $S\subset\{1,\dots,N\}$. We only receive a partial output signal $[\vy_t]_S\triangleq \{y_{t}(k)|k\in S\}$. However, according to \eqref{stack}, the observed signals are influenced by the hidden nodes.

In network topology inference, the goal is to estimate the partial combination matrix $\mA_{S}$ and recover the partial adjacency matrix $\mG_S$ by the data monitored $[\vy_t]_S$. 

\section{Main Results}
This section derives an unbiased estimator for $\mA_S$ and a new node pair clustering based on the GMM method. Finally, we present the NIG algorithm to infer the network topology.
\subsection{An Unbiased Estimator of Matrix $\mA_S$}
First, by summing up $\vx_t$ in \eqref{stack} from 0 to $t$ with the assumption $\vy_0=\mu\vx_0$, $x_{0}(k)\sim  \mathcal N (0,1)$ for any $k$, we obtain
\begin{equation}\label{sum}
	\vy_{t} = \mu\sum_{j=0}^t \mA^{t-j}\vx_{j}.
\end{equation}
Multiplying both sides of equation \eqref{stack} by
	$\vy_{t}^{T}$ and taking the expectations, we obtain
\begin{equation}\label{VAR_change}
	\mathbb{E}[\vy_{t+1}\vy_{t}^T] = \mA\mathbb{E}[\vy_{t}\vy_{t}^T]+\mu\mathbb{E}[\vx_{t+1}\vy_{t}^T].
\end{equation}
Since $\vx_j$, $j=0,\ldots,t+1,$ are i.i.d., according to \eqref{sum}, we obtain $\mathbb{E}[\vx_{t+1}\vy_{t}^T]=0.$
Then equation \eqref{VAR_change} is equivalent to
\begin{equation}\label{AR}
	\mR_1(t)=\mA\mR_0(t),
\end{equation}
where  $\mR_1(t)=\mathbb{E}[\vy_{t+1}\vy_{t}^T]$ and $\mR_0(t)=\mathbb{E}[\vy_{t}\vy_{t}^T].$
Substituting (\ref{sum}) into $\mR_0(t)$ yields
\begin{align}
	\nonumber\mR_0(t)&=\mu^2\mathbb{E}\Big[\Big(\sum_{j=0}^{t} \mA^{t-j}\vx_{j}\Big)\Big(\sum_{j=0}^{t} \mA^{t-j}\vx_{j}\Big)^T\Big]\\
	\label{eqn: R_0}&=\mu^2(\mI+\mA^{2}+\dots+\mA^{2t}),
\end{align}
where the last equality comes from the i.i.d. property and unit variance of the random variables $\vx_{t}$.
Using (\ref{AR}), we obtain
\begin{align}
	\label{eqn: R_1}\mR_1(t) = \mu^2(\mA+\mA^{3}+\dots+\mA^{2t+1}).
\end{align}
\begin{remark}
	Inspired by the residual estimator in \cite{matta2019graph} and using \eqref{eqn: R_0}-\eqref{eqn: R_1}, we can nullify
	the error arising from higher-order powers to obtain the combination matrix $\mA$ by subtraction between $\mR_1(t)$ and $\mA^3\mR_0(t-1)$, which helps us to design an unbiased estimator for the partial combination matrix $\mA_S$.
\end{remark}
Since $\vy_{t}$ is an $N\times1$ vector for any $t$, we propose an $N\times N$ estimator of the matrix $\mA$ at time $t$:
\begin{equation}\label{un}
\mA^u(t) = \frac{1}{\mu^2}(\vy_{t+1}\vy_{t}^T-\vy_{t+2}\vy_{t-1}^T),
\end{equation}
and denote
$[\mA^{u}(t)]_S$ as the submatrix of $\mA^{u}(t)$, which is
\begin{align}
\nonumber[\mA^u(t)]_S &= \frac{1}{\mu^2}\Big([\vy_{t+1}\vy_{t}^T]_S-[\vy_{t+2}\vy_{t-1}^T]_S\Big)\\
\label{uns}&= \frac{1}{\mu^2}\Big([\vy_{t+1}]_S[\vy_{t}^T]_S-[\vy_{t+2}]_S[\vy_{t-1}^T]_S\Big),
\end{align}
where the last equality uses the property of the block matrix product.

Equation \eqref{uns} reveals that the estimator $[\mA^u(t)]_S$ can be obtained by the observed signals.
	It is worth noting that the performance of the proposed estimator also depends on the
	hidden nodes, see Remark 2. 
Next, we show that the proposed estimator benefits from the following property.
\begin{proposition}\label{pro}
	The proposed matrix estimator $\mA^u(t)$ denoted in \eqref{un} is an unbiased estimator of matrix $\mA$. Moreover, $[\mA^u(t)]_S$ is an unbiased estimator of matrix $\mA_S.$
\end{proposition}
\begin{proof}
	We set $\mR_3(t)=\mathbb{E}[\vy_{t+3}\vy_{t}^T]$. Using the definition of $\mR_1(t)$ in \eqref{VAR_change}, we obtain
	\begin{align}\label{eqn: EA}
		\mathbb{E}[\mA^u(t)] = \frac{1}{\mu^2}\big(\mR_1(t)-\mR_3(t-1)\big).
	\end{align}
	Additionally, using the recursion of \eqref{stack}, we obtain
	\begin{align*}
		\vy_{t+3} = \mA^3\vy_{t}+\mu \mA^2\vx_{t+1}+\mu \mA\vx_{t+2}+\mu\vx_{t+3}.
	\end{align*}
Thus,
$\mR_3(t) =  \mA^3\mathbb{E}[\vy_{t}\vy_{t}^T]=\mA^3\mR_0(t),$
where the first equality uses the i.i.d. property. Using \eqref{eqn: R_0}, we obtain
\begin{align}
	\label{eqn: R3_i_1}\mR_3(t-1)= \mu^2(\mA^{3}+\mA^{5}+\dots+\mA^{2t+1}).
\end{align}
Substituting \eqref{eqn: R_1} and \eqref{eqn: R3_i_1} into \eqref{eqn: EA}, we obtain
$
\mathbb{E}[\mA^u(t)] = \mA.
$

This proves that the estimator $\mA^u(t)$ is an unbiased estimator of $\mA$. Considering the partial observation with \eqref{uns}, we have
\begin{equation}\label{unbiased_S}
	\mathbb{E}\big[[\mA^{u}(t)]_S\big]=\frac{1}{\mu^2}\Big([\mR_1(t)]_{S}-[\mR_3(t-1)]_S\Big)=\mA_S.
\end{equation} Thus, $[\mA^u(t)]_S$ is an unbiased matrix estimator of $\mA_S.$
\end{proof}

\begin{remark}\label{remark2}
	The performance of the proposed estimator $[\mA^u(t)]_S$ not only depends on the observed signals but also on the hidden nodes since the hidden nodes influence the observed signals. Moreover, we also showed the hidden nodes do impact the variance of the proposed unbiased estimator in the supplementary material section.
\end{remark}

Proposition \ref{pro} shows that we obtain an unbiased estimator for $\mA_S$ at each time $t$. However, this estimator may have a large variance. To overcome this problem, we propose another unbiased matrix estimator that converges in probability to $\mA_S$. Based on the unbiased estimator and the cumulative samples, we denote
	$
	[\widehat{\mR}_1(n)]_S\triangleq\frac{1}{n}\sum_{t=1}^{n}[\vy_{t+1}]_S[\vy_t]_S^T,$ $[\widehat{\mR}_3(n-1)]_S\triangleq\frac{1}{n}\sum_{t=1}^{n}[\vy_{t+2}]_S[\vy_{t-1}]_S^T.
	$
Then, an \emph{efficient matrix estimator} using cumulative samples is proposed as
\begin{align}\label{eqn: hat_A_s}
	[\widehat{\mA}_n^{u}]_S =\frac{1}{\mu^2}\Big([\widehat{\mR}_1(n)]_S-[\widehat{\mR}_3(n-1)]_S\Big).
\end{align}
Thus, $[\widehat{\mA}_n^{u}]_S $ is the average of $[\mA^u(t)]_S$ from time $1$ to $n$.
Since $[\mA^u(t)]_S$ is an unbiased estimator of $\mA_S$ from Proposition \ref{pro}, $[\widehat{\mA}_n^{u}]_S $ is also unbiased. Besides, it has the following property.

\begin{theorem}\label{theorem}
	Under Assumption \ref{Assu1}, if the observed signals are collected by the VAR model with Gaussian noise in \eqref{va1}, then the unbiased estimator $[\widehat{\mA}_n^{u}]_S$ for the matrix $\mA_S$ satisfies 
	$[\widehat{\mA}_n^{u}]_S\stackrel{P}{\rightarrow} \mA_S,$ when $n\rightarrow\infty.$
\end{theorem}
\begin{proof}
	For brevity, the detailed proof is available in the supplementary material.
\end{proof}

Theorem \ref{theorem} demonstrates that the unbiased estimator $[\widehat{\mA}_n^{u}]_S$ converges to $\mA_S$ in probability.
However, for a fixed network size $N$, the estimators proposed in \cite{matta2019graph} cannot guarantee convergence to $\mA_S$ as the number of samples grows.
In comparison, the advantage of the Granger estimator is that it does not rely on the symmetry assumption. Therefore, it can be used in non-symmetric scenarios.

The detailed comparisons are shown in Table \ref{table:esti}.
\begin{table}[h]
	\caption{Comparison of different estimators}
	\label{table:esti}
	\centering
	\setlength\tabcolsep{5pt}
	\begin{tabular}[1]{cccccccc}
		\hline
		Estimator& Expression& Bias&Convergence  ({*})\\
		\hline
		Granger \cite{matta2019graph}& $[\widehat{\mR}_1(n)]_S([\widehat{\mR}_0(n)]_S)^{-1}$& biased&\ding{53}\\
		One-lag \cite{matta2019graph}& $[\widehat{\mR}_1(n)]_S$& biased&\ding{53}\\
		Residual \cite{matta2019graph}& $[\widehat{\mR}_1(n)]_S-[\widehat{\mR}_0(n)]_S$& biased&\ding{53}\\
		\textbf{Proposed}& $[\widehat{\mR}_1(n)]_S-[\widehat{\mR}_3(n-1)]_S$& unbiased&\ding{52}\\
		\hline
	\end{tabular}
	\begin{flushleft}
		where $[\widehat{\mR}_0(n)]_S\triangleq\frac{1}{n}\sum_{t=1}^{n}[\vy_{t}]_S[\vy_t]_S^T$. $(*)$ means that for a fixed network size $N,$ the estimator converges in probability to $\mA_S$ when $n\rightarrow\infty$.
	\end{flushleft}
\end{table}

\subsection{Node Pair Clustering}
We have developed an unbiased estimator $[\widehat{\mA}_n^{u}]_S $, as denoted in \eqref{eqn: hat_A_s}, that
quantifies the strength of the connections among the network nodes. Considering the binary property of the adjacency matrix, we require the standard clustering algorithms, which allow grouping the entries of the unbiased estimator into two clusters, indicating the ‘disconnection’ and ‘connection’ pairs of the observed network.

To select a suitable clustering method, we complete experiments on ER models to visualize the distribution of the entries of the unbiased matrix estimator $[\widehat{\mA}_n^{u}]_S$ in Fig. \ref{sca}.
For convenience, the `disconnection' and `connection' clusters are denoted as `0' and `1', respectively. $\footnote{In this experiment, we set the network size $N=200$, the connection probability of ER model $p=0.1$, the fraction of the observed nodes $\xi=0.2$ ($S=40$), the parameter of the Laplacian $\lambda = 0.99$ and the weight $\mu=0.1$.}$

\begin{figure}[t]
	\centering
	\subfigure[$N=200,\ \xi=0.2,\ n=10^4$.]{
		\begin{minipage}[t]{0.3\textwidth}
			\centering
			\includegraphics[width=\textwidth]{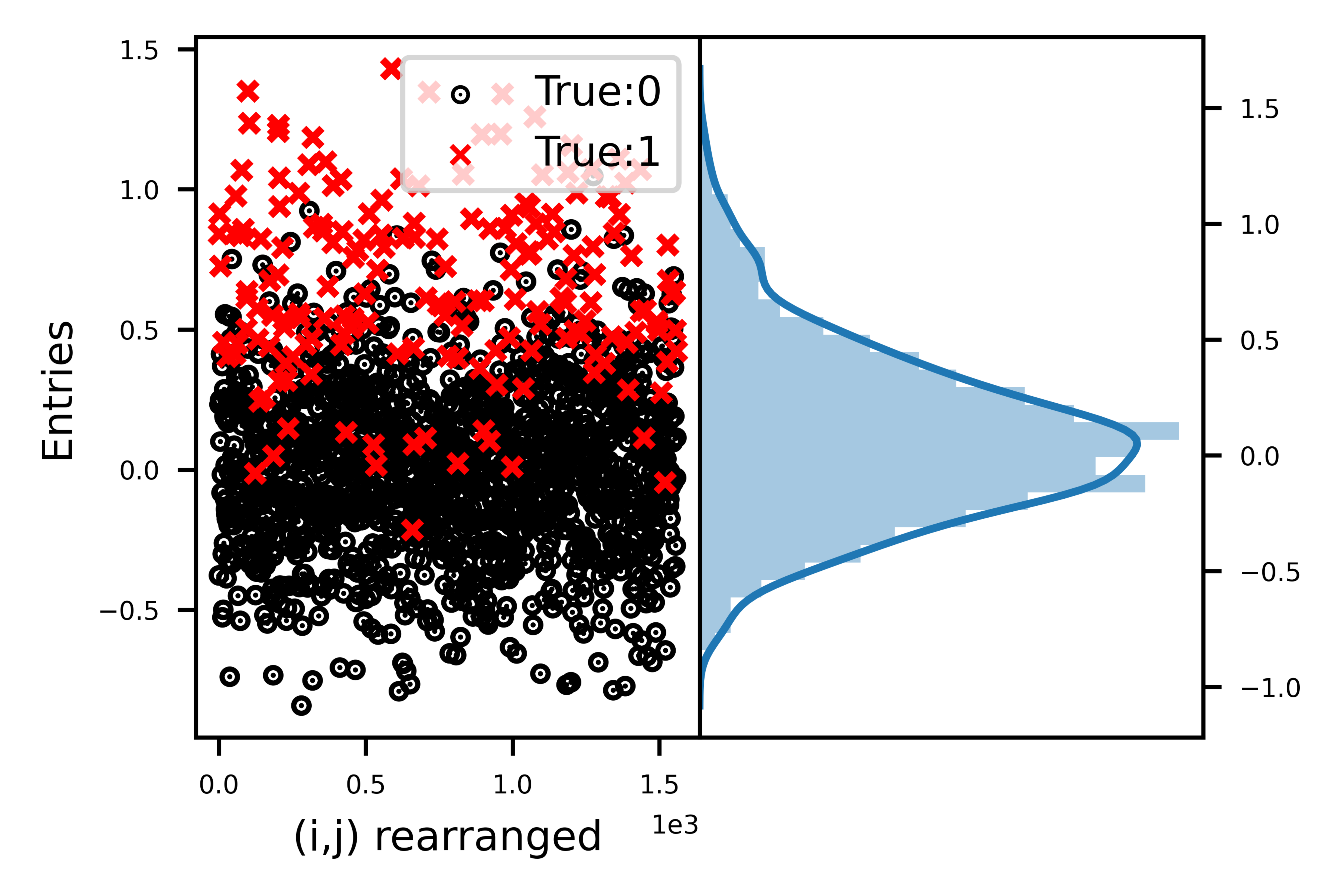}
		\end{minipage}%
	}%
	\subfigure[$N=200,\ \xi=0.2,\ n=10^5$.]{
		\begin{minipage}[t]{0.3\textwidth}
			\centering
			\includegraphics[width=\textwidth]{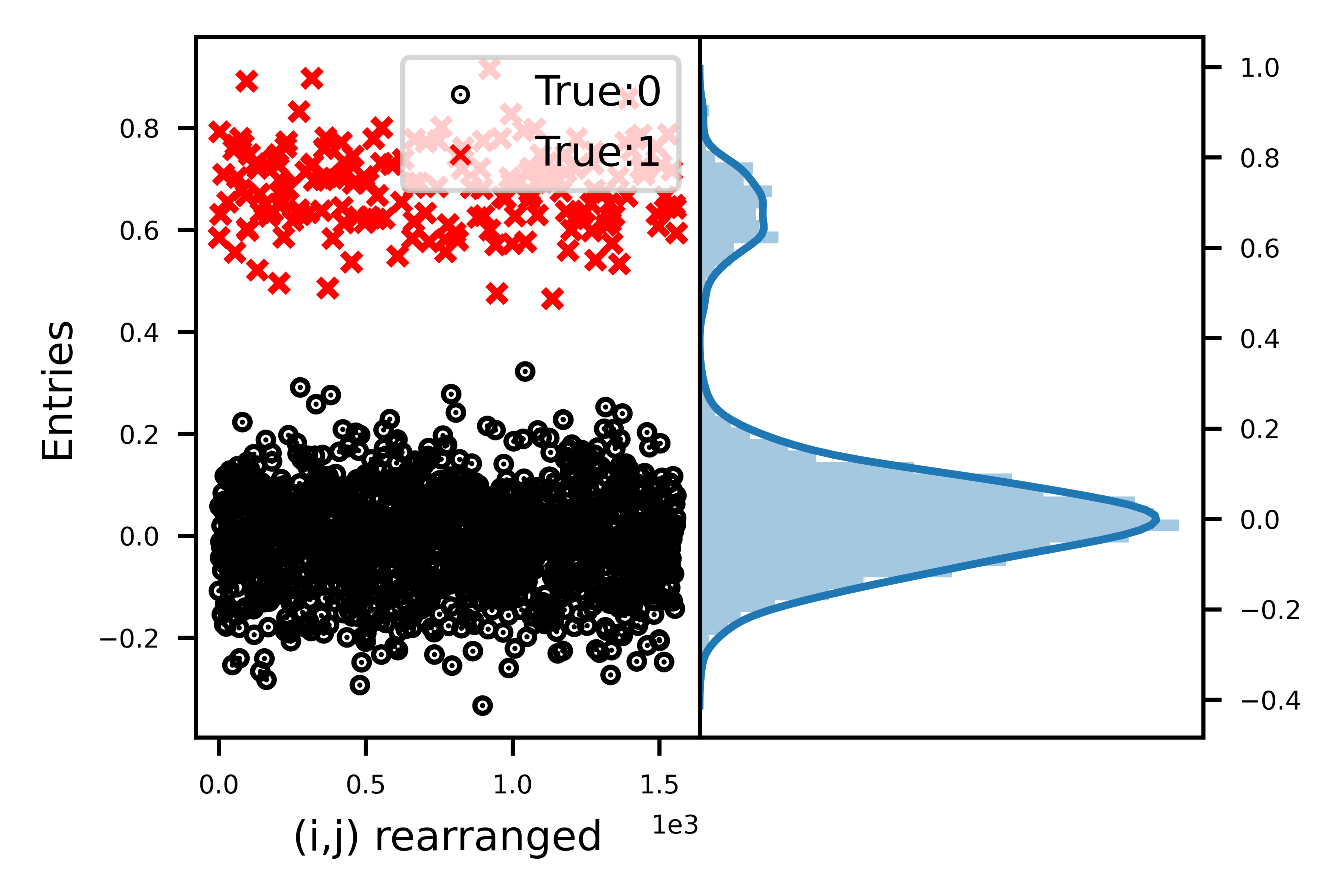}
		\end{minipage}%
	}%
	\centering
	\caption{
		Scatter plots and histograms of the entries of
		the unbiased matrix estimator.}
	\label{sca}
\end{figure}

Fig. \ref{sca} (a) shows that entries are largely mixed under the small size situation. In addition, their distributions also partly overlap. Thus, it is difficult to infer the network topology when the sample size is small. However, when the sample size $n$ is large, as shown in Fig. \ref{sca} (b), the entries are perfectly separated into two clusters, which is consistent with Theorem 1. The histograms in Fig. \ref{sca} reveal that their distribution seems to be a Gaussian mixture distribution.

Thus, we adopt the GMM algorithm \cite{najar2017comparison} 
to group the entries into two clusters. It is known that the GMM algorithm is implemented by the expectation-maximum algorithm, which is sensitive to the initial conditions \cite{biernacki2003choosing}.
Since the value for the disconnected pairs under the Laplacian rule
is 0, the expectation of the unbiased estimator for the disconnected pairs is also 0 (see Fig. \ref{sca}).
Thus, we set 0 as the initial mean for the disconnected pairs when using the GMM algorithm.

\subsection{Proposed Method}
The proposed NIG method is summarized in Algorithm \ref{algorithm}, which is used to infer network topology under partial observability.
\begin{algorithm}[h]\label{algorithm}
	
	\caption{NIG algorithm}
	\KwIn{Streaming data $[\vy_0]_S,\dots,[\vy_{n+2}]_S$; Weight $\mu$;}
	\KwOut{Prediction matrix $\mG_S$}
	
	$[\widehat{\mR}_1(n)]_S\leftarrow\frac{1}{n}\sum_{t=1}^{n}[\vy_{t+1}]_S[\vy_t]_S^T$
	
	$[\widehat{\mR}_3(n-1)]_S\leftarrow\frac{1}{n}\sum_{t=1}^{n}[\vy_{t+2}]_S[\vy_{t-1}]_S^T$
	
	$[\widehat{\mA}_n^{u}]_S\leftarrow\frac{1}{\mu^2}([\widehat{\mR}_1(n)]_S-[\widehat{\mR}_3(n-1)]_S)$
	
	$\mG_S\leftarrow$ GMM$([\widehat{\mA}_n^{u}]_S)$~~~\# $\textmd{initial mean: 0.}$
	
	
\end{algorithm}
Steps 1-3 are completed to generate the unbiased estimator $[\widehat{\mA}_n^{u}]_S$. Step 4 is completed using the GMM clustering method for the entries of the unbiased matrix estimator.

\section{Numerical Experiments}

\begin{figure*}[htbp]
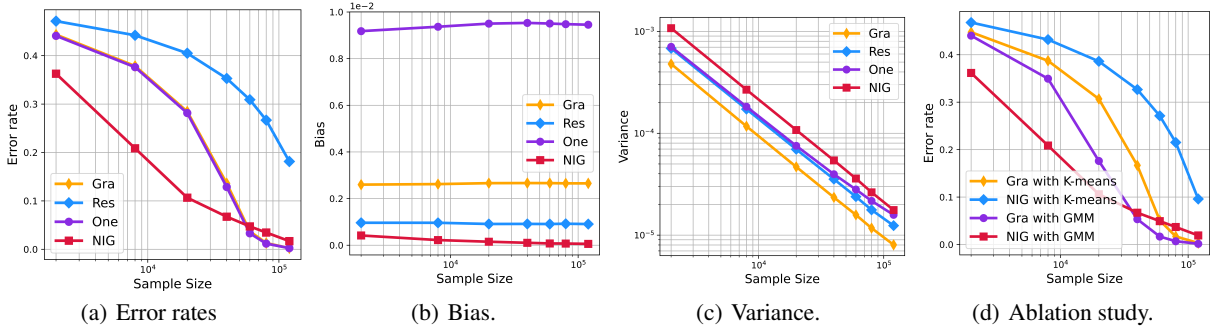

	\centering
	\subfigure[Error rates]{
		\begin{minipage}[t]{0.24\textwidth}
			\centering
			\includegraphics[width=\textwidth]{figure/AQ-error-PLUS-eps-converted-to}
		\end{minipage}%
	}%
	\subfigure[Bias.]{
		\begin{minipage}[t]{0.24\textwidth}
			\centering
			\includegraphics[width=\textwidth]{figure/AQ-bias-eps-converted-to}
		\end{minipage}%
	}%
	\subfigure[Variance.]{
		\begin{minipage}[t]{0.24\textwidth}
			\centering
			\includegraphics[width=\textwidth]{figure/AQ-variance-eps-converted-to}
		\end{minipage}%
	}%
	\subfigure[Ablation study.]{
		\begin{minipage}[t]{0.24\textwidth}
			\centering
			\includegraphics[width=\textwidth]{figure/AQ-ablation-eps-converted-to}
		\end{minipage}%
	}
	\centering
	\caption{Comparison of the proposed NIG algorithm with other methods.}
	\label{bvar}
\end{figure*}
\begin{figure*}[htbp]
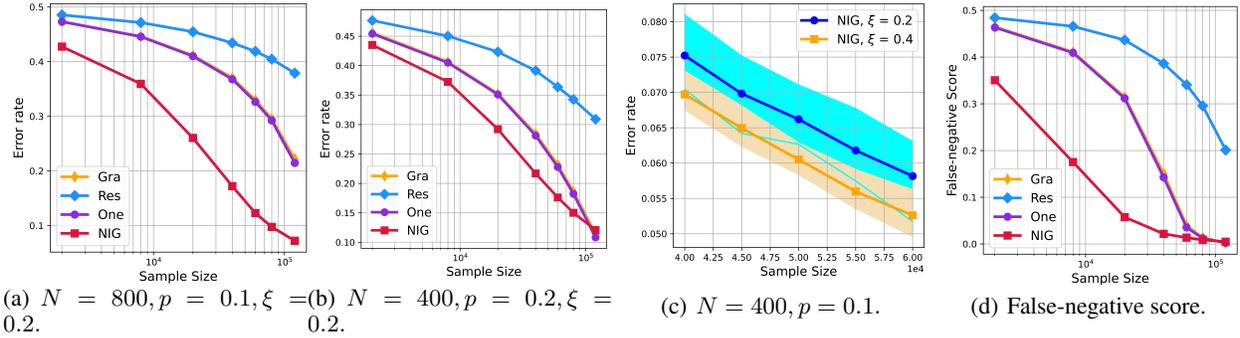

	\centering
	\subfigure[$N=800, p=0.1, \xi=0.2.$]{
		\begin{minipage}[t]{0.24\textwidth}
			\centering
			\includegraphics[width=\textwidth]{figure/AQ-error8-plus-eps-converted-to}
		\end{minipage}%
	}%
	\subfigure[$N=400, p=0.2, \xi=0.2.$]{
		\begin{minipage}[t]{0.24\textwidth}
			\centering
			\includegraphics[width=\textwidth]{figure/AQ-error2-eps-converted-to}
		\end{minipage}%
	}%
	\subfigure[$N=400, p=0.1.$]{
		\begin{minipage}[t]{0.25\textwidth}
			\centering
			\includegraphics[width=\textwidth]{figure/AQ_zhixin_gauss1_mu1-eps-converted-to}
		\end{minipage}%
	}
	\subfigure[False-negative score.]{
		\begin{minipage}[t]{0.24\textwidth}
			\centering
			\includegraphics[width=\textwidth]{figure/AQ-FN-PLUS-eps-converted-to}
		\end{minipage}%
	}
	\centering
	\caption{Comparison of the proposed NIG algorithm under different settings.}
	\label{error}
\end{figure*}
In this section, we examine the numerical performance of
the proposed NIG algorithm and present comparison
results with three existing methods, namely, one-lag (One), residual (Res), and Granger (Gra) estimators followed by the K-means clustering algorithm in \cite{matta2019graph}. If not otherwise specified, the NIG algorithm uses the GMM clustering algorithm. The error rates are adopted for performance evaluation, which are the fractions of network node pairs that are incorrectly recovered. 

Consider a network with $N$ nodes, where, according to the ER model, edges are determined by independent Bernoulli distributions with link probability $p$.
If not otherwise specified, we set $N=400, \xi = 0.2, \mu = 0.1, p=0.1.$
The parameter $\lambda$ of the Laplacian rule is set as 0.99. To assess the performance of the proposed unbiased estimator in the NIG algorithm,
we consider the following scenarios, where the results are averaged over 50 experiments.
\begin{itemize}
	\item 
	The proposed algorithm is compared with other methods in terms of error rate, bias, and variance in Fig. \ref{bvar}.
	\item 
	An ablation study is provided to clarify where the gain of the proposed algorithm arises. Specifically, in Fig. \ref{bvar} (d),
	we show the results from the K-means and GMM algorithms
	for the proposed unbiased estimator and Granger estimator, respectively.
	\item We examine the error rates and 95$\%$ confidence intervals of the NIG algorithm when varying $N$, $p$, and $\xi$ in Fig. \ref{error}. Moreover, in Fig. \ref{error} (d), we compare these algorithms in terms of a popular metric, the false-negative score \cite{sokolova2006beyond}, which is defined as the fraction of disconnected node pairs declared as connected.
\end{itemize}
Fig. \ref{bvar} (a) shows that the proposed estimator has the lowest error rate compared to the other estimators in the case of small sample sizes. However, when the sample size becomes larger than $0.8*10^5$, the proposed algorithm has higher error rates than the Granger and one-lag estimators. The reason may be that
the proposed estimator has the smallest bias, which mainly affects the performance and makes the proposed unbiased estimator perform best in the small sample size situation. However, with the gradual increase in sample sizes, the variance becomes the main factor affecting the performance. Thus, the larger variance leads to the higher error rates of the proposed NIG algorithm than baselines (see Fig. \ref{bvar} (b)-(c)). In addition, one can observe that its variance converges to 0 as $n$ increases, which is consistent with Theorem \ref{theorem}.

Fig. \ref{bvar} (d) shows that both the proposed and Granger estimators with the GMM algorithm outperform those with the K-means algorithm. Thus, the GMM algorithm provides some gains for both estimators. When the sample size is small, one can also observe that the proposed estimator has a better performance than the Granger estimator with the same GMM algorithm. The reason is that the unbiased estimator provides a better initialization for the GMM clustering algorithm.
Thus, the gains of the proposed method arise from both the new estimator and the GMM algorithm.

Comparing the proposed method with the baselines in Fig. \ref{error} (a) and Fig. \ref{bvar} (a), we see that the performance gap grows with larger network sizes. Fig. \ref{error} (b) shows that all the estimators have higher error rates when the graph becomes denser, but the proposed estimator still performs best 
with small sample sizes.
Fig.\ref{error} (c) shows that, with a higher fraction of observed nodes, the proposed NIG algorithm has a narrower confidence interval and a lower average error rate.
	The reason may be that the impact of the hidden nodes on the observed signals decreases when the fraction of the observed nodes increases since there are fewer hidden nodes. In addition, Fig. \ref{error} (d) shows the false-negative score. A satisfactory estimator should have small false-negative scores. Fig. \ref{error} (d) shows that the proposed method has the lowest false-negative score and demonstrates that the proposed method has the highest probability of recovering the disconnected pairs correctly. On the other hand, the proposed NIG algorithm has contrary performance under the false-positive score \cite{sokolova2006beyond}.

\section{Conclusion}
In this letter, the novel NIG algorithm, including an unbiased estimator and a GMM algorithm for network topology inference under partial observability, is proposed. Theoretically, we have proven that the unbiased estimator converges to the network combination matrix in probability and helps us obtain a better initialization for the GMM algorithm. The simulations have shown that the proposed NIG algorithm obtains a better performance with small sample sizes, sparse networks, and high fractions of observed nodes. Future work aims to improve the algorithms by combining the advantages of the biased estimators and the proposed method.

\bibliographystyle{ieeetr}
\bibliography{bare_jrnl1021}

	{\setcounter{equation}{0}
		\renewcommand{\theequation}{S\arabic{equation}}
		
		{\setcounter{equation}{0}
			\renewcommand{\theequation}{S\arabic{equation}}

			\renewcommand{\thetable}{S\arabic{table}}
			\renewcommand{\thefigure}{S\arabic{figure}}
			
\section{Supplementary Material}
			
			\textbf{Notations:} 
				Define $(\mM)_v\triangleq\textmd{vec}(\mM)$ to be the vectorization of a matrix $\mM$. The covariance between the two random matrices $\mM_1$ and $\mM_2$ is denoted as $\textmd{Cov}(\mM_1, \mM_2) \triangleq \textmd{Cov}[(\mM_1)_v, (\mM_2)_v]= \mathbb{E}[(\mM_1-\mathbb{E}[\mM_1])_v(\mM_2-\mathbb{E}[\mM_2])_v^T]$ \cite{christensen2015covariance} \cite[eq. (305)]{petersen2008matrix}. If $\mM_1=\mM_2$, the covariance $\textmd{Cov}(\mM_1, \mM_1)$ is abbreviated as $\textmd{Cov}(\mM_1)$. Moreover, $\textmd{Tr}[\textmd{Cov}(\mM)]=\sum_{i,j}\textmd{Var}(M_{ij}),$ which explains that the trace of covariance equals the sum of all entries' variances. The spectral radius of the combination matrix $\mA$ is denoted as $\rho(\mA).$ Finally, $n$ denotes the number of samples on each node. $\textbf{0}$ denotes the zero matrices.}
		\subsection{Proof of Theorem 1}
		Let us define the combination matrix estimator $\widehat{\mA}_n^{u}\triangleq \frac{1}{\mu^2}[\widehat{\mR}_1(n)-\widehat{\mR}_3(n-1)],$
		where $\widehat{\mR}_3(n-1)\triangleq\frac{1}{n}\sum_{t=1}^{n}\vy_{t+2}\vy_{t-1}^T=\frac{1}{n}\sum_{t=0}^{n-1}\vy_{t+3}\vy_{t}^T$ and $\widehat{\mR}_1(n)\triangleq\frac{1}{n}\sum_{t=1}^{n}\vy_{t+1}\vy_t^T$. Then, 
		using the fact that $\textmd{Var}(A\pm B)\leq2(\textmd{Var}(A)+\textmd{Var}(B)),$
		we have
		\begin{equation}\label{inequality}
			\textmd{Tr}[\textmd{Cov}(\mM_1+\mM_2)]\leq2(\textmd{Tr}[\textmd{Cov}(\mM_1)]+\textmd{Tr}[\textmd{Cov}(\mM_2)]).
		\end{equation}
	Then, applying \eqref{inequality} on $\widehat{\mA}_n^{u}$, we have
		\begin{equation}\label{eqn:ineq}	0\leq\textmd{Tr}[\textmd{Cov}(\widehat{\mA}_n^{u})]\leq \frac{2}{\mu^4}\textmd{Tr}\Big[\textmd{Cov}\Big(\widehat{\mR}_1(n)\Big)\Big]+
			\frac{2}{\mu^4}\textmd{Tr}\Big[\textmd{Cov}\Big(\widehat{\mR}_3(n-1)\Big)\Big],
		\end{equation}
		where $\textmd{Tr}$ is the trace of a matrix. From the VAR model in (2), we have $\vy_{t+j} = \mA^j\vy_{t}+\mu(\sum_{k=1}^{j}\mA^{j-k}\vx_{t+k})$, which is  multiplied on both sides by
		$\vy_{t}^{T}$, yielding
		$
		\mathbf{y}_{t+j} \mathbf{y}_{t}^{T}=\mathbf{A}^{j} \mathbf{y}_{t} \mathbf{y}_{t}^{T}+\mu \sum_{k=1}^{j} \mathbf{A}^{j-k} \mathbf{x}_{t+k} \mathbf{y}_{t}^{T}, \forall j \geq 1.$ Considering $\mM_1+ \mM_2 + \mM_3 + \mM_4$, using \eqref{inequality}, we have
		\begin{align}\label{8}
			&\textmd{Tr}[\textmd{Cov}(\mM_1+\mM_2+ \mM_3 + \mM_4)] \leq 2\textmd{Tr}[\textmd{Cov}(\mM_1)]+2\textmd{Tr}[\textmd{Cov}(\mM_2+ \mM_3 + \mM_4)]\\
			\nonumber&\leq 2\textmd{Tr}[\textmd{Cov}(\mM_1)]+4\textmd{Tr}[\textmd{Cov}(\mM_2)]+ 4\textmd{Tr}[\textmd{Cov}(\mM_3 + \mM_4)]\\
			\nonumber&\leq 2\textmd{Tr}[\textmd{Cov}(\mM_1)]+4\textmd{Tr}[\textmd{Cov}(\mM_2)]+ 8\textmd{Tr}[\textmd{Cov}(\mM_3)] + 8\textmd{Tr}[\textmd{Cov}(\mM_4)]\\
			\nonumber&\leq 2\textmd{Tr}[\textmd{Cov}(\mM_1)]+8\sum_{i=2}^{4}\textmd{Tr}[\textmd{Cov}(\mM_i)].
		\end{align}
		
		Based on the definitions of $\widehat{\mR}_1(n)$, $\widehat{\mR}_3(n-1)$, \eqref{inequality}, and \eqref{8} we have
		\begin{align}
			\label{R3in}&\textmd{Tr}[\textmd{Cov}(\widehat{\mR}_1(n))]+\textmd{Tr}[\textmd{Cov}(\widehat{\mR}_3(n-1))]\\
			\nonumber&\leq 2\textmd{Tr}\Big[\textmd{Cov}\Big(\frac{1}{n}\sum_{t=1}^{n}\mA \vy_t\vy_t^T\Big)\Big]
			+8\mu^2\sum_{k=1}^{3}\textmd{Tr}\Big[\textmd{Cov}\Big(\frac{1}{n}
			\sum_{t=0}^{n-1}\mA^{3-k}\vx_{t+k}\vy_t^T\Big)\Big]
			\\
			\nonumber&+2\textmd{Tr}\Big[\textmd{Cov}\Big(\frac{1}{n}\sum_{t=0}^{n-1}\mA^3 \vy_t\vy_t^T\Big)\Big] +2\mu^2\textmd{Tr}\Big[\textmd{Cov}\Big(\frac{1}{n}\sum_{t=1}^{n}\vx_{t+1}\vy_t^T\Big)\Big].
		\end{align}
		Combining \eqref{R3in} with \eqref{eqn:ineq}, we derive an upper bound for $\textmd{Tr}[\textmd{Cov}(\widehat{\mA}_n^{u})]$. In Section \ref{sec_1} and Section \ref{sec_2}, we have proved $\textmd{Tr}[\textmd{Cov}(\frac{1}{n}\sum_{t=1}^{n}\vy_{t}\vy_t^T)]\rightarrow0$ and	$\textmd{Tr}[\textmd{Cov}(\frac{1}{n}\sum_{t=1}^{n}\vx_{t+j}\vy_t^T)]\rightarrow0$, yielding $\textmd{Tr}[\textmd{Cov}(\widehat{\mA}_n^{u})]\rightarrow0.$
		Since $\textmd{Tr}[\textmd{Cov}([\widehat{\mA}_n^{u}]_S)]$ is the sum of the partial diagonal entries of $\textmd{Cov}(\widehat{\mA}_n^{u})$,
		we have $\textmd{Tr}[\textmd{Cov}([\widehat{\mA}_n^{u}]_S)]\rightarrow0$.
		
		Since $[\widehat{\mA}_n^{u}]_S$ is an unbiased estimator of $\mA_S,$ 
		using the multidimensional Chebyshev's inequality \cite{ferentios1982tcebycheff}, we have
		\begin{equation}\label{eqn:prob_S}
			1\geq P\Big(\Vert[\widehat{\mA}_n^{u}]_S - \mA_S\Vert_F^2<\epsilon^2\Big)\geq 1 - \frac{\textmd{Tr}[\textmd{Cov}([\widehat{\mA}_n^{u}]_S)]}{\epsilon^2},
		\end{equation}
		for any $\epsilon>0.$ Then, substituting  $\textmd{Tr}[\textmd{Cov}([\widehat{\mA}_n^{u}]_S)]\rightarrow0$ into \eqref{eqn:prob_S}, we have $\lim\limits_{n\rightarrow\infty}P(\Vert([\widehat{\mA}_n^{u}]_S - \mA_S)\Vert_F<\epsilon)=1.$
		Therefore, we derive $[\widehat{\mA}_n^{u}]_S\stackrel{P}{\rightarrow} \mA_S$ when $n\rightarrow\infty.$
		\subsection{Proof of ${\textmd{Tr}}[{\textmd{Cov}}(\frac{1}{n}\sum_{t=1}^{n}\vx_{t+j}\vy_t^T)]\rightarrow0$, $j\geq1.$}\label{sec_1}
		Since $\textmd{Cov}(\frac{1}{n}\sum_{t=1}^{n}\vx_{t+j}\vy_t^T) = \frac{1}{n^2}\sum_{t=1}^{n}[\textmd{Cov}(\vx_{t+j}\vy_t^T)+\sum_{k\neq t}^{n}\textmd{Cov}(\vx_{k+j}\vy_k^T,\vx_{t+j}\vy_t^T)],$ we first consider the trace of the following covariance,
		\begin{align}
			&\nonumber\textmd{Tr}\Big[\textmd{Cov}\Big(\vx_{k+j}\vy_k^T,\vx_{t+j}\vy_t^T\Big)\Big]=\textmd{Tr}\Big[\mathbb{E}[(\vx_{k+j}\vy_k^T)_v(\vx_{t+j}\vy_t^T)^T_v]\Big]\\
			&\nonumber=\mathbb{E}\Big[\textmd{Tr}[(\vx_{t+j}\vy_t^T)^T_v(\vx_{k+j}\vy_k^T)_v]\Big]= \mathbb{E}\Big[\textmd{ Tr}[\textmd{Tr}[\vy_{t}\vx_{t+j}^T\vx_{k+j}\vy_{k}^T]]\Big]\\
			&\label{COV_total}=\mathbb{E}\Big[\textmd{Tr}[\vy_{t}\vx_{t+j}^T\vx_{k+j}\vy_{k}^T]\Big],
		\end{align}
		where the first equality comes from the definition of covariance and uses the fact that $\vx_t$ is an (i.i.d.) zero-mean random variable with $\mathbb{E}[\vx_{t+j}\vy_t^T]=0$, the second equality uses the property of trace in \cite[eq. (14)]{petersen2008matrix} and the third equality uses the property of vectorization, i.e., $(\mB)_v^T(\mC)_v=\textmd{Tr}(\mB^T\mC)$ \cite[eq. (521)]{petersen2008matrix}.
		
		When $k\neq t$, using the i.i.d. property of $\vx_t$ and (\ref{COV_total}), we have $\textmd{Tr}[\textmd{Cov}(\vx_{k+j}\vy_k^T,\vx_{t+j}\vy_t^T)]=0.$ When $k=t$, equation (\ref{COV_total}) yields	
		\begin{align}
			\textmd{Tr}[\textmd{Cov}(\vx_{t+j}\vy_t^T)]=\textmd{Tr}[\mathbb{E}[\vx^T_{t+j}\vx_{t+j}]\mathbb{E}[\vy^T_t\vy_t]].
		\end{align}
		where the equation uses the properties of trace, i.e.,  $\textmd{Tr}(\mB\mC)=\textmd{Tr}(\mC\mB)$ in \cite[eq. (14)]{petersen2008matrix}. Since $\vx_t\sim N(0,\mI_N),$ $\vx^T_t\vx_t$ follows a chi-square distribution with $N$ degrees of freedom and $\mathbb{E}[\vx^T_t\vx_t]=N$. Then, we have
		\begin{align}
			\textmd{Tr}[\textmd{Cov}(\vx_{t+j}\vy_t^T)]=N\textmd{Tr}[\mathbb{E}\vy_t\vy^T_t]=N\textmd{Tr}[\mR_0(t)]\label{lemma:var2},
		\end{align}
		where the last equation comes from $\vy_t\sim N(0,\mR_0(t))$ in equation (3).
		Finally, combining \eqref{lemma:var2} with \eqref{COV_total}, we obtain
		\begin{equation}\label{trace}		\textmd{Tr}\Big[\textmd{Cov}\Big(\frac{1}{n}\sum_{t=1}^{n}\vx_{t+j}\vy_t^T\Big)\Big]=\frac{N}{n^2}\textmd{Tr}\Big[\sum_{t=1}^{n}\mR_0(t)\Big].
		\end{equation}
		Using the Stolz Theorem \cite{schmeelk2013elementary}, we have
		$\lim\limits_{n \to \infty}\frac{1}{n}\sum_{t=1}^{n}\mR_0(t)=\lim\limits_{n \to \infty}\mR_0(n).$ Recall that  $\mR_0(t)=\mu^2(\mI+\mA^{2}+\dots+\mA^{2t})$, yielding that $\lim\limits_{n \to \infty}\mR_0(n)=\lim\limits_{n \to \infty}\mu^2\sum_{t=0}^{n}\mA^{2t}=\mu^2(\mI-\mA^2)^{-1}$, where the last equality comes from  $\rho(\mA)<1$ and \cite[eq. (487)]{petersen2008matrix}.
	 Thus, we obtain
	$\lim\limits_{n \to \infty}\frac{1}{n}\sum_{t=1}^{n}\mR_0(t)=\mu^2(\mI-\mA^2)^{-1}.$
	Substituting this into \eqref{trace}, we have $\textmd{Tr}[\textmd{Cov}(\frac{1}{n}\sum_{t=1}^{n}\vx_{t+j}\vy_t^T)]\rightarrow0,$ when $n\rightarrow\infty.$
		\subsection{Proof of ${\textmd{Tr}}[{\textmd{Cov}}(\frac{1}{n}\sum_{t=1}^{n}\vy_{t}\vy_t^T)]\rightarrow0$.}\label{sec_2}
		According to the VAR model (2), we have $\vy_j\vy_j^T=[\mA^{j-t}\vy_{t}+\mu\sum_{k=1}^{j-t}\mA^{j-t-k}\vx_{t+k}][\mA^{j-t}\vy_{t}+\mu\sum_{k=1}^{j-t}\mA^{j-t-k}\vx_{t+k}]^T,$ for $t<j.$ Then, let us first consider
		\begin{align}\label{step2}
			~~	\nonumber\textmd{Tr}\Big[\textmd{Cov}\Big(\vy_j\vy_j^T,\vy_t\vy_t^T\Big)\Big]
			&=\textmd{Tr}\Big[\textmd{Cov}\Big(\mA^{j-t}\vy_t\vy_t^T\mA^{j-t},\vy_t\vy_t^T\Big)\Big]\\
			&=\textmd{Tr}\Big[\mA^{j-t}\otimes\mA^{j-t}\textmd{Cov}\Big(\vy_t\vy_t^T\Big)\Big],
		\end{align}
		where the first equality uses \eqref{COV_total} and the last equality uses property
		$(\mA^{j-t}\vy_t\vy_t^T\mA^{j-t})_v=\mA^{j-t}\otimes\mA^{j-t}(\vy_t\vy_t^T)_v$ in \cite[eq. (520)]{petersen2008matrix}, where $\otimes$ means the Kronecker product.
		Since $\vy_t\sim N(0,\mR_0(t))$, the streaming data $\vy_t\vy_t^T$ follows a Wishart distribution. Then we have $\mathbb{E}[\vy_t\vy_t^T]=\mR_0(t)$ and $\textmd{Cov}(\vy_t\vy_t^T)=2\mR_0(t)\otimes \mR_0(t)$ in  \cite[Proposition 8.3]{eaton2007wishart}. Substituting this into \eqref{step2}, we obtain
		$\textmd{Cov}\Big(\vy_j\vy_j^T,\vy_t\vy_t^T\Big)=2\mA^{j-t}\otimes\mA^{j-t}(\mR_0(t)\otimes \mR_0(t))$. Thus, we have
			\begin{align}		    \nonumber&\textmd{Tr}\Big[\textmd{Cov}(\frac{1}{n}\sum_{t=1}^{n}\vy_{t}\vy_t^T)\Big]=\frac{1}{n^2}\sum_{t=1}^{n}\textmd{Tr}\Big[\textmd{Cov}(\vy_{t}\vy_t^T)+2\sum_{j> t}^{n}\textmd{Cov}(\vy_{j}\vy_j^T,\vy_{t}\vy_t^T)\Big]\\
				&=\textmd{Tr}\Big[\frac{2}{n^2}\sum_{t=1}^{n}\Big(1+2\sum_{j> t}^{n}\mA^{j-t}\otimes\mA^{j-t}\Big)\Big(\mR_0(t)\otimes\mR_0(t)\Big)\Big].\label{pei}
			\end{align}
		Then, we apply the corollary of the Stolz theorem \cite[Example 3.75]{schmeelk2013elementary}, which says that two convergent sequences $a_n$, $b_n$ satisfy that $\lim\limits_{n\rightarrow\infty}a_n=a$, $\lim\limits_{n\rightarrow\infty}b_n=b$, then, 
		\begin{equation}\label{C4}
			\lim\limits_{n\rightarrow\infty}\frac{a_1b_n+a_2b_{n-1}+\dots+a_nb_1}{n}=ab.
		\end{equation}
		One may notice that $a_{i+1}, b_{n-i} \in \mathbb{R}.$ However, \eqref{C4} also holds when $\left(a_{n}\right)_{n \in \mathbf{N}}$ and $\left(b_{n}\right)_{n \in \mathbf{N}}$ are matrix sequences by just using a matrix norm (for example, the Frobenius norm) to replace the absolute value in the proof of \cite[Example 3.75]{schmeelk2013elementary}. We provide the proof for the matrix version here:
		
		Let two matrix sequences $\left(A_{n}\right)_{n \in \mathbf{N}}$ and $\left(B_{n}\right)_{n \in \mathbf{N}}$ be given and let us define the sequence $\left(C_{n}\right)_{n \in \mathbf{N}}$ by
		$$
		C_{n}=\frac{A_{1} B_{n}+A_{2} B_{n-1}+\cdots+A_{n} B_{1}}{n}, \quad n=1,2, \ldots
		$$
		(a) If $\lim _{n \rightarrow \infty} A_{n}=\textbf{0}$ and $\Vert B_{n}\Vert_F \leq b$ for every $n \in \mathbf{N}$, then it holds $\lim _{n \rightarrow \infty} C_{n}=\textbf{0}$.
		
		(b) If $\lim _{n \rightarrow \infty} A_{n}=A$ and $\lim _{n \rightarrow \infty} B_{n}=B$, then it holds $\lim _{n \rightarrow \infty} C_{n}=AB$.
		
		Proof (a): From
		$$
		\lim _{n \rightarrow \infty} A_{n}=\textbf{0} \Rightarrow \lim _{n \rightarrow \infty}\Vert A_{n}\Vert_F=0,
		$$
		and the Stolz theorem, it follows
		$$
		\lim _{n \rightarrow \infty} \frac{1}{n} \sum_{k=1}^{n}\Vert A_{k}\Vert_F=\lim _{n \rightarrow \infty}\Vert A_{n}\Vert_F=0.
		$$
		Since the sequence $\left(B_{n}\right)_{n \in \mathbf{N}}$ is bounded, i.e., $\Vert B_{n}\Vert_F \leq b$ for every $n \in \mathbf{N}$, using the sub-additive and sub-multiplicative properties of matrix norm, i.e., $\Vert \mM_1+\mM_2\Vert_F\leq \Vert \mM_1\Vert_F + \Vert \mM_2\Vert_F$, $\Vert \mM_1\mM_2\Vert_F\leq \Vert \mM_1\Vert_F\Vert \mM_2\Vert_F$, we have $\Vert C_{n} \Vert_F \leq \frac{\Vert A_{1}\Vert_F+\Vert A_{2}\Vert_F+\cdots+\Vert A_{n}\Vert_F}{n} b$ and $\lim_{n \rightarrow \infty}\Vert C_{n}\Vert_F \leq \lim _{n \rightarrow \infty} \frac{b}{n} \sum_{k=1}^{n}\Vert A_{k}\Vert_F=0.$ This means that $\lim _{n \rightarrow \infty} C_{n}=\textbf{0}$.
		
		Proof (b): Let us put
		$\lim _{n \rightarrow \infty} A_{n}=a$ and $X_{n}=A_{n}-A$ for every $n \in \mathbf{N} .$
		Then it holds $\lim _{n \rightarrow \infty} X_{n}=\textbf{0}$ and
		$$
		\begin{aligned}
			C_{n} &=\frac{\left(X_{1}+A\right) B_{n}+\cdots+\left(X_{n}+A\right) B_{1}}{n} \\
			&=\frac{X_{1} B_{n}+\cdots+X_{n} B_{1}}{n}+A \cdot \frac{B_{n}+\cdots+B_{1}}{n}=F_{n}+G_{n} .
		\end{aligned}
		$$ Since $\lim _{n \rightarrow \infty} B_{n}=B$, the sequence $\left(B_{n}\right)_{n \in \mathbf{N}}$ is bounded, according to (a), it follows that
		$$
		\lim _{n \rightarrow \infty} F_{n}=\lim _{n \rightarrow \infty} \frac{X_{1} B_{n}+\cdots+X_{n} B_{1}}{n}=\textbf{0}
		$$
		while
		$$
		\lim _{n \rightarrow \infty} G_{n}=\lim _{n \rightarrow \infty} A \cdot \frac{B_{n}+\cdots+B_{1}}{n}=AB.
		$$
		Therefore $\lim _{n \rightarrow \infty} C_{n}=AB$. Then we have completed the proof of the corollary of matrix version.

		Back to the proof of \ref{sec_2}, we can set that $$B_{n-t}= \frac{1+2\sum_{j> t}^{n}\mA^{j-t}\otimes\mA^{j-t}}{n}=\frac{1+2\mA\otimes\mA+\dots+2\mA^{n-t}\otimes\mA^{n-t}}{n}.$$ and $$A_{t+1} = \mR_0(t)\otimes\mR_0(t).$$ Then, the matrix in trace operation on the right side of \eqref{pei} has the shape of \eqref{C4}. Moreover, using the Stolz Theorem and $\rho(\mA)<1,$ we have $\lim\limits_{n \to \infty} B_{n} = \lim\limits_{n \to \infty}\frac{1}{n}(1+2\sum_{j> 0}^{n}\mA^{j}\otimes\mA^{j})=\lim\limits_{n \to \infty}2\mA^{n}\otimes\mA^{n}=\textbf{0},$ and $\lim\limits_{n \to \infty}A_n = \lim\limits_{n \to \infty}\mR_0(n)\otimes\mR_0(n)=\mu^4(\mI-\mA^2)^{-1}\otimes(\mI-\mA^2)^{-1}$. Then according to the corollary of the Stolz theorem mentioned above, we have
		\begin{equation}\label{conclusion}
			\lim\limits_{n \to \infty}\textmd{Tr}\Big[\textmd{Cov}(\frac{1}{n}\sum_{t=1}^{n}\vy_{t}\vy_t^T)\Big]= 0.
		\end{equation}
		
	\subsection{Proof of Remark \ref{remark2}: the hidden nodes impact the upper bound of the proposed unbiased estimator's variance.}
	
	
	Considering the variance of the unbiased estimator, using \eqref{inequality} and \eqref{uns}, we obtain the upper bound of the variance as follows,
	\begin{align}\label{eqn1}
		\textmd{Tr}[\textmd{Cov}([\mA^u(t)]_S)]\leq \frac{2}{\mu^4}\textmd{Tr}\Big[\textmd{Cov}\Big([\vy_{t+1}\vy_{t}^T]_S\Big)\Big]+
		\frac{2}{\mu^4}\textmd{Tr}\Big[\textmd{Cov}\Big([\vy_{t+2}\vy_{t-1}^T]_S\Big)\Big].
	\end{align}
	
	Multiplying both sides of equation \eqref{stack} by
	$\vy_{t}^{T}$ and  considering the partial observability, we obtain
	\begin{align*}
		[\mathbf{y}_{t+1} \mathbf{y}_{t}^{T}]_S=[\mathbf{A} \mathbf{y}_{t} \mathbf{y}_{t}^{T}]_S+\mu [\mathbf{x}_{t+1} \mathbf{y}_{t}^{T}]_S.
	\end{align*}
	Then, according to \eqref{inequality}, we have
	\begin{align}\label{eqn2}
		\textmd{Tr}[\textmd{Cov}([\mathbf{y}_{t+1} \mathbf{y}_{t}^{T}]_S)]\leq 2\textmd{Tr}\Big[\textmd{Cov}\Big([\mA \vy_t\vy_t^T]_S\Big)\Big]+
		2\mu^2\textmd{Tr}\Big[\textmd{Cov}\Big([\vx_{t+1}\vy_t^T]_S\Big)\Big].
	\end{align}
	Similarly to \cite[eq. (62)]{matta2018consistent}, with rules for multiplication between partitioned matrices, we have
	\begin{align}\label{eqn3}
		[\mA \vy_t\vy_t^T]_S= [\mA]_S [\vy_t\vy_t^T]_S + [\mA]_{SS'} [\vy_t\vy_t^T]_{S'S},
	\end{align}
	where $S'$ is defined  as the set of hidden nodes. $[\mA]_{SS'}$ denotes the submatrix spanning the rows, indexed by the set $S$, and columns, indexed by the set $S'$ of the combination matrix $\mA$. According to \eqref{inequality}, we have
	\begin{equation}\label{3.6}
		\textmd{Tr}\Big[\textmd{Cov}\Big([\mA \vy_t\vy_t^T]_S\Big)\Big]\leq2\textmd{Tr}\Big[\textmd{Cov}\Big([\mA]_S [\vy_t\vy_t^T]_S\Big)\Big] + 2 \textmd{Tr}\Big[\textmd{Cov}\Big([\mA]_{SS'} [\vy_t\vy_t^T]_{S'S}\Big)\Big].
	\end{equation}
	Combing \eqref{eqn1}, \eqref{eqn2}, and \eqref{3.6}, we show that the upper bound of the variance of the proposed estimator is influenced by the hidden nodes and the missing observations.
	\end{document}